\documentclass[11pt,letterpaper]{article}

\usepackage{theorem}
\usepackage{amsmath}
\usepackage{definitions}
\usepackage{ifthen}
\usepackage{xspace}
\usepackage{amsmath,amsfonts,graphpap,amscd,mathrsfs,graphicx,lscape}
\usepackage{epsfig,amssymb,amstext,xspace}

\usepackage{bbm}
\usepackage{algorithm,algorithmic}
\usepackage{tikz, url}

\usepackage{color}              
\usepackage{epsfig}

\usepackage{authblk}

\newcommand{\comment}[1]{}

\newcommand{\vc}[1]{\ensuremath{\mathbf{#1}}\xspace}

\newcommand{\one}{\mathbbm{1}}

\newcommand{\abs}[1]{\vert{#1}\vert}

\newcommand{\rev}{\mathcal{R}}

\newcommand{\nash}{\mathcal{N}}
\newcommand{\prob}{\mathbb{P}}
\newcommand{\E}{\mathbb{E}}

\newcommand{\Ts}{\ensuremath{\vc{T}}\xspace}
\newcommand{\TP}{\ensuremath{\vc{T'}}\xspace}

\renewcommand{\PROB}{\ensuremath{\prob}\xspace}
\renewcommand{\Expect}[2][]{\ensuremath{%
\ifthenelse{\equal{#1}{}}{\E}{\E_{#1}}%
[#2]}\xspace}

\begin{document}

\begin{titlepage}

\title{Pricing Public Goods for Private Sale}

\author[1]{Michal Feldman}
\author[2]{David Kempe}
\author[3]{Brendan Lucier}
\author[3]{Renato Paes Leme}
\affil[1]{Hebrew University and Harvard University}
\affil[2]{University of Southern California}
\affil[3]{Microsoft Research}

\renewcommand\Authands{ and }

\date{}

\maketitle

\begin{abstract}
We consider the pricing problem faced by a seller who assigns a price
to a good that confers its benefits not only to its buyers, but also to other
individuals around them.
For example, a snow-blower is potentially useful not only to the
household that buys it, but also to others on the same street.
Given that the seller is constrained to selling such a (locally) public
good via individual private sales, how should he set his prices
given the distribution of values held by the agents?

We study this problem as a two-stage game.  In the first stage,
the seller chooses and announces a price for the product. In
the second stage, the agents (each having a private value for
the good) decide simultaneously whether
or not they will buy the product.
In the resulting game, which can exhibit a multiplicity of equilibria, 
agents must strategize about whether they will
themselves purchase the good to receive its benefits.

In the case of a fully public good (where all agents benefit
whenever any agent purchases), 
we describe a pricing mechanism that is approximately
revenue-optimal (up to a constant factor) when values are
drawn from a regular distribution.
We then study settings in which the good is only ``locally'' public:
agents are arranged in
a network and share benefits only with their neighbors.
We describe a pricing method that approximately maximizes
revenue, in the worst case over equilibria of agent behavior,
for any $d$-regular network.  
Finally, we show that
approximately optimal prices can be found for general networks
in the special case that private values are drawn from a
uniform distribution.
We also discuss some barriers to extending these results
to general networks and regular distributions.

\end{abstract}

\end{titlepage}

\section{Introduction} \label{sec:intro}

Pricing products for sale is an important strategic decision for
firms. Based on the demand at different prices, an optimal price
should maximize the number of items sold, times the revenue per sold
item. A long history of work in economics, and more recently in
computer science, studies the problem of finding an optimal price (or,
more generally, selling mechanism), given a demand curve or estimate
thereof \cite{myerson-81, milgrom2004putting,krishna2009auction, hartline_lectures12}.

This view ignores the fact that products frequently exhibit
externalities: if a consumer $j$ purchases the product, it may
affect the utility of consumer $i$. These externalities naturally
differ in two dimensions: (1) whether they are positive or negative,
and (2) whether they affect other consumers when they purchase the
product, or when they do not purchase it.

Some of the classical literature in economics
\cite{jehiel:moldovanu:interdependent,jehiel:moldovanu:externalities,jehiel:moldovanu:stacchetti:nuclear,jehiel:moldovanu:stacchetti:multidimensional,brocas:endogenous}
focuses on negative externalities experienced by consumer $i$ as a
result of $j$'s purchase, regardless of whether $i$ himself purchases.
Motivating examples are weapons or powerful competitive technologies.
If a competitor $j$ has access to these technologies, it poses an
often significant threat to $i$, and hence, $i$ would be willing to
pay significant amounts of money to prevent $j$ from acquiring the
product.
There has been a recent focus on positive externalities between pairs
$i,j$ when both purchase
\cite{hartline:mirrokni:sundarajan,arthur:motwani:sharma:xu,akhlaghpour10,anari10,Haghpanah11,Bhalgat12}.
This type of scenario arises, for instance, for implicit
creation of technology standards, where the use of a particular
technology (such as an operating system or cell phone plan) becomes
more advantageous as others use the same technology.
In this context, the focus is often on finding the right ``seeds'' to
create enough implicit peer influence effects; de facto, some users
are offered much lower prices to serve as seeds.

In the present paper, we investigate important domains of
externalities, and the impact they have on pricing decisions.
Our main focus is on
\emph{positive externalities from purchasers on non-purchasers}.
In other words, when one customer purchases an item, others will
derive utility from it, even if they themselves do not purchase it.
This is the case commonly known as \emph{public goods} in
economics
\cite{samuelson54,Bergstrom86,mas-collel:whinston:green}.
Public goods arise in many real-world scenarios:
\begin{enumerate}
\item If one researcher acquires a useful piece of infrastructure
  (such as a poster printer), other research groups in the same
  department profit as well.
\item If one family purchases a useful and expensive  gardening tool, its
  neighbors can borrow the tool and use it as well.
\item {If a company finances useful infrastructure in a region,
  it also makes the region more attractive for other companies.
  One concrete} example is the Wi-Fi networks that Google
  recently built in Chelsea and in Kansas City
  \cite{google_internet_chelsea}, {which are expected to attract more
  talent to those areas}.
\end{enumerate}


Since the goods described above benefit an entire group of agents, one way of
purchasing them would be to gather as a group, purchase a single copy,
and split its cost among the group members.
This is, however, not always possible due to various reasons:
in case (1), regulations might allow a researcher to pay for a printer
from his grant budget, but not to pay for it partially;
in case (2), the family might consider it impolite to ask each
potential borrower of the gardening tool to contribute to it;
and in case (3), the companies that will benefit from the
infrastructure being in place might be competitors and therefore might
be unwilling to cooperate.
{More fundamentally, it has been long known that rational
agents in these types of settings have incentives to misrepresent
their true utilities (see, e.g., \cite{samuelson54}), a phenomenon
colloquially known as free-riding.}
Hence, it is very common that, despite the public-good nature of these
goods, purchases are made privately;
that is, one agent purchases the good, incurring the entire cost
alone, while benefiting the group as a whole.
It is crucial for a seller who is offering the product for sale to
take these externality considerations into
account.\footnote{The examples listed above can be considered nearly
\emph{pure} public goods, in that the benefits from being the purchaser
and being a ``neighbor'' are very similar. A much larger number of
products --- such as most entertainment technology --- has a
significant public component, but also a significant private
component. We discuss this interesting extension as a direction
for future work in Section~\ref{sec:concl}.}
Overall, we would expect the demand for such items to be 
reduced given that the buyers, taken as a whole, will demand fewer
copies.

We model the locally public nature of the good as follows.
We consider a graph $G$ that captures the interactions between the buyers.
Each buyer has a non-negative valuation drawn independently
from a distribution $F$ common to all buyers.
If the buyer or one of his neighbors purchases the good, he obtains his
valuation as utility; if he was the one purchasing it, then the good's
price (set by the seller) is subtracted from his utility.
We study the Nash Equilibria of the game described above and the
problem faced by the seller of setting a price (based on the graph and the
valuation distribution) in order to optimize the revenue at equilibrium.

In investigating this question, we are interested in understanding the
influence of the different parameters on the optimal pricing choice.
For example, how does the optimal price depend on the topology of the
network $G$?
Since it is usually hard for the seller to learn the buyers' social
network, is it possible to find a price that will generate
approximately optimal revenue for any network;
or, {if not, a price
that depends only on simple statistics about
the network, such as its average degree?}
We are also interested in investigating the power of discriminatory
vs.~non-discriminatory pricing.
Can the seller benefit from setting a different price for each agent?
Is there a non-discriminatory price that gives a good approximation
with respect to every discriminatory pricing policy?

\paragraph{Negative externalities and the Hipster Game} Our framework can
also be used to study other types of externalities. Consider a product
that serves the role of a fashion statement or status symbol. In that case,
it may be essential to the purchaser to be the \emph{only} one with
a copy. His utility is, therefore, his valuation if he has the product and no
other agent in his neighborhood has it. Otherwise, his utility is zero. We call
this the \emph{Hipster Game}. We show that the pricing problem in the
Hipster Game is analogous to the problem for public goods, and
the same algorithmic and analysis techniques yield essentially
identical results.

\subsection{Our Contribution}
%


\paragraph{Globally public goods}
We begin our study by focusing on the complete graph, i.e., the case
of \emph{globally public goods}.
We are interested in prices which will yield high revenue at
equilibrium. One immediate obstacle in this context is that the
(Bayesian) purchasing game played by the buyers may have (infinitely)
many equilibria.
We show that nonetheless, there is a single price $p$ which can be
computed explicitly from the agents' value distribution, and
which is approximately optimal in the following very strong sense:
The revenue under the worst-case equilibrium at price $p$ is within
a constant factor of the revenue of the \emph{best} equilibrium
for the best general (not necessarily uniform) price vector.
In other words, price discrimination can improve revenue by at most a
constant factor, even if one is optimistic about the equilibrium that
will be reached.

Our analysis draws a relation between our problem and the optimal
(Myerson) revenue of a single-item auction among $n$ bidders.
The main insight driving our result is that, at equilibrium, the
agents aim to make purchasing decisions so that only one
agent will buy the product, in expectation.
This connection allows us to leverage the rich literature on
single-item auctions for our analysis; it also
explicates the connection to the Hipster Game, where positive utility
can only be derived when exactly one agent obtains the good.


\paragraph{Locally public goods}
With a solid understanding of globally public goods in place,
we next turn our attention to \emph{locally public goods}, which are
modeled by arbitrary networks $G$. At this point, we cannot answer the
question of finding optimal prices for arbitrary $G$.
However, we make significant progress on the question, as follows.

First, we consider the case of $d$-regular graphs $G$.
Here, the results on globally public goods carry over in spirit.
However, technically, the assertion is weaker: we show how to
explicitly compute a uniform price $p$ which, when offered to all the
agents, is guaranteed to obtain a constant fraction of the
\emph{worst-case revenue} for any fixed price $p'$.
Remarkably, this price depends only on the degree $d$ and the
distribution $F$, and is independent of the actual graph structure.
Notice that the guarantee is weaker than the one for globally public goods
in two respects: (1) it only provides guarantees compared to one fixed
price, not a price vector with discrimination, and
(2) it compares only to the worst-case revenue for these other prices
$p'$ (instead of the best-case one).
This weaker assertion is inevitable: we show that there exist
$d$-regular graphs in which the gap between the best worst-case
revenue and the best revenue in equilibrium is $\Theta(d)$, and similarly, the
gap between the worst-case revenue of the best uniform price
vector and the best discriminating prices is $\Theta(d)$ as well.

We next consider the case of general graphs.
We present evidence that our previous approaches will face
inherent difficulties in handling general graphs.
In particular, we give an instance of a network such that, for every
price, the gap between the best-case and worst-case revenue is
$\Theta(n)$.
Therefore, approximate optimality of worst-case equilibria cannot be
established by bounding best-case revenues.
At a minimum, this raises an equilibrium selection problem: which is
the right revenue to optimize, and to compare against?

For $d$-regular graphs, our solution concept is to bound worst-case 
revenue for the price against
the worst-case revenue at other prices $p'$.
We show that for general graphs, this approach faces a fundamental
obstacle: approximating worst-case revenue to within a factor
$n^{1-\epsilon}$ for a given price is NP-hard, even if $F$ is the
uniform distribution. Thus, we do not expect a concise or useful
characterization of the approximate worst-case revenue.

Surprisingly, for the specific case of the uniform distribution $F$,
one does not need to be able to compute the objective function in
order to \emph{optimize} it: for the uniform distribution $F$, simply
offering a price of $\frac{1}{2}$ guarantees worst-case revenue within
a factor at most $4/e$ of optimal.
Unfortunately, the analysis techniques for this case rely very
specifically on the uniform distribution of valuations; it is an
interesting open question whether they can be extended beyond the
uniform distribution.

\subsection*{Related Work}
Externalities in general, and public goods in particular,
have a rich and long history of study in economics.
The tension arising from private provisioning of public goods has been
realized since the early studies of public goods:
Samuelson \cite{samuelson54} already noticed that private
provisioning will not necessarily achieve a social optimum.
(See also the discussion in Chapter 11 of \cite{mas-collel:whinston:green}.)
Implicit in the study of markets for public goods in this literature
is the goal of setting the right price, taking into account
production costs and utility curves.
Our model differs from the classic models in that purchase decisions
are binary, whereas traditional models allow agents to choose a
continuous level $x_i$ at which to purchase the public good.
Each agent's utility in the fully public setting is a function of
$\sum_i x_i$, whereas interpreting the $x_i$ as probabilities, the
utilities in our setting are of the form $1-\prod_i (1-x_i)$.
Thus, the analysis techniques commonly used in the literature on
public goods do not apply directly in our setting.

The study of private sales of public goods is also
present in the classic paper of Bergstrom, Blume and Varian
\cite{Bergstrom86} and in work by Allouch
\cite{Allouch}. The authors consider a model in which
agents need to split an initial endowment of public and private goods.
The focus of those papers is to prove existence and uniqueness of
equilibria in such games.

In our work, we assume that the good to be allocated is fully public.
There is a large body of literature studying the effects of
\emph{congestion}, where a good's value to an individual
decreases as others use it. Several works study
allocation mechanisms to price such congestion effects,
going back to the original work of
Pigou~\cite{pigou}. For 
overviews of
pricing of congestion in public and club goods, see
\cite{cornes:sandler:externalities,jackson:nicolo}.
In the present work, the good does not become congested;
instead, a graph structure specifies which individuals derive utility
from the purchases made by others.

A study of locally public goods in the graph-theoretic sense
considered here\footnote{Past work on ``local public goods'' used the
term to describe public goods for a community such as a small
town. As such, the term corresponds to a fully public (though
possibly congestible) good, when the set of individuals under
consideration is restricted.}
has only been begun much more recently, as part of the recent trend
toward studying classic games in a networked setting.
(See Galleotti \cite{Galeotti08} for a general overview.)
Specifically, locally public goods have been studied by
Bramoull\'{e} and Kranton \cite{BramoulleKranton07} and Bramoull\'{e},
Kranton and D'Amours \cite{BramoulleKranton10}.
Bramoull\'{e} and
Kranton~\cite{BramoulleKranton07,BramoulleKranton10} study a
setting in which agents decide on a level of effort; an agent's
utility grows as a function of the cumulative efforts of himself and
all his neighbors in the network. In this sense, the model generalizes
the classical public goods model to networks; 
as we discussed above, in contrast, our model focuses on
  probabilistic decisions to purchase or not to purchase.
One main difference between \cite{BramoulleKranton07} and our work is
that, instead of merely taking the games as given, we seek to engineer
the network game by setting parameters (in our case: prices) that will
lead to more desirable equilibria (equilibria of higher revenue).

Also closely related to our model is the work of Candogan, Bimpikis and
Ozdaglar~\cite{candogan}. This work considers a monopolist
who sets prices for agents that are embedded in a network and exhibit
positive externalities. 
Their model differs from ours in three main respects. 
First, as with the work of Bramoull\'{e} et al., the level of
  consumption in their model is continuous rather than binary.
Second, their externality model is different in that an agent's
utility is \emph{additive} over the purchases made by his neighbors,  
whereas in our case, purchases of neighbors are substitutes.
Third, they adopt a full-information model, in which the auctioneer
knows the demands of the agents, whereas in our model, the agents'
values are drawn from a known distribution.

We focus our attention on mechanisms that allocate a (globally or locally)
public good by way of \emph{posted prices}.  Posted price mechanisms
have received significant recent attention in the context of auctions
with multiple objects for sale \cite{chawla07,chawla10}
where it has been shown that, in various settings,
approximately optimal revenue can be extracted by offering
a vector of take-it-or-leave-it prices to each buyer in sequence.  Our
analysis shares similarities with this line of work: like \cite{chawla07},
we relate our pricing problem to a corresponding single-item auction
problem.
However, unlike \cite{chawla07},
setting a posted price in an auction for a \emph{public} good can lead
to multiple equilibria of buyer behavior,
with different equilibria generating substantially
different revenues.  

\section{Models and Preliminaries} \label{sec:models}

We write $[n] = \SET{1,2,\ldots,n}$.
Throughout, vectors are denoted by bold face.
The buyers form the vertices $V = [n]$ of an undirected graph
$G = ([n], E)$.
The neighbors of a node $i \in V$ are denoted by
$N(i) = \Set{j}{(i,j) \in E}$, with the convention
that $i \notin N(i)$.
For an event $\mathcal{E}$, we write $\one\{\mathcal{E}\}$ for the
function whose value is 1 when $\mathcal{E}$ happens and 0 otherwise.

We are interested in \emph{locally public goods}: goods that let a
player derive utility either from being allocated the good, or from
having a neighbor who is allocated the good.
More formally, we define utilities as follows:
Each agent $i$ has a private valuation $v_i$ for the good, drawn
independently from a common and commonly known
atomless distribution $F$.\footnote{Some of our preliminary results
carry over to the case when buyers have different
distributions $F_i$.} Since we assume that $F$ is atomless, for every $q \in
[0,1]$, there is at least one value of $p$ for which $F(p) = q$.
We write $F^{-1}(q) = \min \{p \vert F(p) = q\}$.

If $S$ is the set of agents allocated the good, and
$\pi_i$ the payment of agent $i$, then agent $i$'s utility is
\begin{align*}
u_i (S, \pi_i)
& = \begin{cases}
v_i - \pi_i & \text{if } i \in S \text{ or } S \cap N(i) \neq \emptyset \\
- \pi_i & \text{otherwise.}
\end{cases}
\end{align*}
A natural question arises regarding whether agents $i \notin S$
should have non-zero payments, given that they may profit from the
allocation to their neighbors.
In the present work, we focus on the \emph{private sale} of the good
via \emph{posted prices}, i.e., the seller determines the price of the
good, and an agent is only charged when purchasing the good.
This is the most widely used mechanism for selling goods, public or
private.

We remark that since our setting is a single-parameter setting,
Myerson's theory of optimal auctions \cite{myerson-81} would
yield a revenue-optimal mechanism.
However, the mechanism does not correspond to private sales 
since it charges not only the buyers, but also their neighbors who
derive benefit from the item.

\subsection{Equilibria in the posted-prices game}
The pricing decisions can be modeled as a two-stage game.
In the first stage, the seller sets a price vector
$\vc{p}$ to offer the buyers.
For most of the paper, and unless specified otherwise,
all agents will be offered the same price $p$. Subsequently, the
buyers play a simultaneous Bayesian game.
The seller's goal is to choose $\vc{p}$ so as to maximize
revenue.

We assume that the agents maximize their expected utility.
Given a price $p_i$, a player $i$ will buy if his utility from buying,
$v_i - p_i$, exceeds the expected utility from not buying,
$v_i \cdot \left( 1 - \prod_{j \in N(i)} \Prob{j \text{ does not buy}} \right)$.
At equality, $i$ could randomize between the two strategies, but
since we assumed the distribution $F$ to be atomless, equality is an
event of probability 0.
Thus, each agent will employ a threshold strategy: buy if and
only if
\begin{align*}
v_i
& \geq \frac{p_i}{\prod_{j \in N(i)} \Prob{j \text{ does not buy}}}
\; =: \; T_i.
\end{align*}
Because all other players $j$ also employ threshold strategies,
we can write
$\Prob{j \text{ does not buy}} = \Prob{v_j \leq T_j} = F(T_j)$.
Thus, the Nash Equilibria are exactly the threshold vectors
$\Ts = (T_1, \hdots, T_n)$ satisfying the following condition:
\begin{align}
T_i \cdot \prod_{j \in N(i)} F (T_j)
& = p_i, \quad \text{for all } i \in V.
\label{eqn:equilibrium-condition}
\end{align}

Given a price vector $\vc{p}$, we use
$\nash_{\vc{p}}$ to denote the set of Nash Equilibria
$\Ts = (T_1, \hdots, T_n)$ of the posted prices game with prices
$\vc{p}$.
We prove below that $\nash_{\vc{p}} \neq \emptyset$.
Given a Nash Equilibrium $\Ts \in \nash_{\vc{p}}$, the
corresponding expected revenue is
$\rev(\vc{p},\Ts) = \sum_i p_i \cdot (1-F(T_i)).$

\subsubsection{Existence of (possibly multiple) Equilibria}
To prove the existence of at least one equilibrium, define
\begin{align*}
B & =
\left[p_1, \frac{p_1}{\prod_{j \in N(1)} F(p_j)}\right] \times
\left[p_2, \frac{p_2}{\prod_{j \in N(2)} F(p_j)}\right] \times
\cdots \times
\left[p_n, \frac{p_n}{\prod_{j \in N(n)} F(p_j)}\right],
\end{align*}
and consider the best-response function
$\Psi: B \rightarrow \R^n_+$, defined as
$\Psi_i(\Ts) = p_i / \prod_{j \in N(i)} F(T_j)$.
We claim that $\Psi(\Ts) \in B$ for all $\Ts \in B$.

First, notice that for any $\Ts$, we have $\Psi_i(\Ts) \geq p_i$.
Intuitively, this captures that, regardless of the other players'
strategies, no player will ever buy the good for more than his value.
On the other hand, because $T_j \geq p_j$ for $\Ts \in B$,
we also get that
\begin{align*}
\Psi_i(\Ts)
& = \frac{p_i}{\prod_{j \in N(i)} F(T_j)}
\; \leq \frac{p_i}{\prod_{j \in N(i)} F(p_j)}.
\end{align*}

Thus, $\Psi: B \to B$ is a continuous function from $B$ to $B$.
So long as the prices are such that $F(p_i) > 0$ for all $i$,
$B$ is compact, and the existence of a fixed point (and thus an
equilibrium of the game) follows from Brouwer's Fixed Point Theorem.

If there is one or more agent $i$ with $F(p_i) = 0$, then the
following construction proves the existence of an equilibrium.
For each agent $i$ with $F(p_i) = 0$, set $T_i = p_i$,
and for all neighbors of $i$, set $T_i = \infty$.
In other words, $i$ deterministically buys the good, and $i$'s
neighbors never buy the good. It follows directly from the definition
of $\Psi$ that the best response for all these agents will be
$\Psi_i (\Ts) = T_i$. Since the agents with $T_i = \infty$
contribute a term $F(\infty) = 1$ to their neighbors' product, the
remaining problem remains unchanged if we remove all these agents
completely, and focus on the restriction of $\Psi$ to the remaining
agents. For those, the previous compactness argument applies.

\begin{remark}
In general, there could be many equilibria of the game.
Even for the special case of $G = K_n$ and $F(x) = x$ for $x \in [0,1]$, any
threshold vector \Ts with
$\prod_i T_i = p$ is a Nash Equilibrium of the posted-prices game with
uniform prices $p$.
Thus, there is in general a continuum of equilibria.
\end{remark}

\subsubsection{Symmetric Equilibria for $d$-Regular Graphs}
When the graph $G$ is $d$-regular, and the prices offered to the
buyers are the same, i.e., $p_i = p$ for all $i$, then we can show
that $\Psi$ also has a symmetric equilibrium.
Notice that if $\Ts = T \cdot \vc{1}$, then
$\Psi(\Ts) = p / F(T)^d \cdot \vc{1}$, so the best responses will be
symmetric.
It therefore suffices to study the function
$\psi(T) = p / F(T)^d$, and show that it has a fixed point.
To see this, observe that the condition for the existence of a symmetric equilibrium is the existence of a threshold
$T$ such that $T \cdot F(T)^d = p$.
Because $\psi(p) = p / F(p)^d$ and
$\psi( p / F(p)^d ) \leq p / F(p)^d$, the existence of a fixed point
in the interval $[p, p / F(p)^d]$
follows by the intermediate value theorem.

\subsection{Hipster Game}

In this section, we consider the following variation of the game.
In the \emph{Hipster Game}, each agent strives to be unique
among his friends, so upon acquiring a good, he only derives value
from it if he is the only person in his social network who has this
good. More precisely, if $S$ is the set of allocated agents, and
$\vc{\pi}$ is the vector of payments, then:
\begin{align*}
u_i (S, \pi_i)
& = \begin{cases}
v_i - \pi_i & \text{if } i \in S \text{ and } S \cap N(i) = \emptyset \\
- \pi_i & \text{otherwise.}
\end{cases}
\end{align*}
Notice that this definition of utilities seems to give us a
game which is the complete opposite of the Public Goods Game.
While the Public Goods Game was an example of positive externalities,
the Hipster Game is an example of negative externalities.
In fact, this game can be described as a congestion game:
the graph nodes are congestable resources, and the resources
requested by a player are exactly all nodes in his neighborhood.
While the Hipster Game is characterized by negative
  externalities, it exhibits a very similar equilibrium structure to
the Public Goods Game.
Player $i$ decides to purchase the good for price $p_i$ if
\begin{align*}
v_i \cdot \Prob{\text{no agent in } N(i) \text{ buys}} - p_i
& \geq 0.
\end{align*}
Therefore, the set of equilibria for this game is composed of threshold
strategies for all agents such that the thesholds satisfy
$p_i = T_i \cdot \prod_{j \in N(i)} F(T_j)$ for all $i$.

Thus, the Public Goods Game and the Hipster Game have the same
set of equilibria and also the same revenue.
(However, they are \emph{not isomorphic}, since the payoff structure is
not the same.)
We can use this observation to get a crude upper bound on the
expected revenue of the Public Goods Game for arbitrary
graphs.
We note that it is equal to the expected revenue of the
Hipster Game, which in turn is at most the expected
welfare of the Hipster Game, as each agent must derive
non-negative expected utility at equilibrium. 
The expected welfare of the Hipster Game is at most the
expected weight of the maximum weighted independent set with
weights $v_i$ drawn i.i.d.~from $F$.
Thus, we conclude:

\begin{lemma}
The expected revenue from the Public Goods Game is at most the expected
weight of the maximum weighted independent set of
$G$ with node weights $v_i$ drawn i.i.d.~from $F$.
\end{lemma}

\subsection{Regularity, Myerson's Lemma and the Prophet Inequality}\label{sec:review}

Much of our analysis will be based on Myerson's Lemma about the
optimal selling mechanism, combined with the prophet inequality.
We briefly review these concepts here.
A more comprehensive exposition can be found in Hartline's lecture
notes \cite{hartline_lectures12}.

\begin{definition}[Virtual values and regularity]
Let $F$ be the cumulative distribution function of an atomless distribution
on an interval $[a,b]$, and let $f$ be its corresponding density function.
The \emph{virtual value function} associated with distribution $F$ is
defined as $\phi(x) = x - \frac{1-F(x)}{f(x)}$.
The distribution $F$ is \emph{regular} if $\phi(x)$ is non-decreasing.
\end{definition}

Consider a single-agent scenario in which an agent with value
$v$, drawn from $F$, is made a take-it-or-leave-it offer at price $p$.
The agent will accept the offer iff his value exceeds $p$, which
happens with probability $1-F(p) =: q$.
Therefore, the revenue obtained by posting a price $p$ is
$p \cdot (1-F(p))$, which can be also written in terms of the
\emph{quantile space} as $q \cdot F^{-1}(1-q)$.
This motivates the following definition:

\begin{definition}[Revenue curve]
The \emph{revenue curve} corresponding to the cumulative distribution
function $F$ is a function $R:[0,1] \rightarrow \R_+$, defined by
$R(q) = q \cdot F^{-1}(1-q)$.
It specifies the revenue as a function of the \emph{ex ante}
probability of sale.
\end{definition}

The derivative of the revenue curve with respect to $q$ is
$\DiffOp{R}{q} (q)	 = F^{-1}(1-q) - \frac{q}{f(F^{-1}(1-q))}
\; = \; \phi(F^{-1}(1-q))$. Since $F^{-1}$ is monotone non-decreasing, a
distribution is
regular iff its corresponding revenue curve is concave.

\subsubsection{Single-item auctions and Myerson's Lemma}
We draw repeatedly on the scenario in which a single item is sold to
$n$ agents with valuations drawn i.i.d.~from a regular distribution $F$.
A mechanism receives a vector of bids
$\vc{b} = (b_1, \hdots, b_n)$ and returns an allocation vector
$\vc{x}(\vc{b}) = (x_1, \ldots, x_n) \in \R^n_+$
such that $\sum_i x_i \leq 1$,
and a payment vector $\vc{\pi}(\vc{b}) \in \R^n_+$.
The mechanism is \emph{incentive compatible} if no bidder can benefit
from reporting a value other than his true value, i.e., if bidding
$b_i = v_i$ is a weakly dominant strategy for each agent $i$.
Myerson~\cite{myerson-81} established the following lemma, which
relates the payments of an incentive compatible mechanism to the
expected virtual values:

\begin{lemma}[Myerson~\cite{myerson-81}]
\label{lem:myerson-lemma}
For any incentive compatible mechanism, and any bidder $i$,
$\Expect[\vc{v}]{\pi_i(\vc{v})}
= \Expect[\vc{v}]{x_i(\vc{v}) \cdot \phi(v_i)}$.
\end{lemma}

In particular, it follows from Lemma~\ref{lem:myerson-lemma} that the
revenue-maximizing incentive compatible mechanism allocates
the item entirely to an agent with highest non-negative virtual value.
The \emph{$n$-agent Myerson Revenue} is the optimal revenue that can
be obtained in a setting with a single item and $n$ agents,
and is given by $\rev^M_n = \Expect{\max_{i} \phi(v_i)^+}$,
where $z^+ = \max (0,z)$.
When clear from the context, we drop the subscript $n$.

Lemma~\ref{lem:myerson-lemma} also implies that the optimal mechanism
for selling an item to a single agent is a posted price mechanism with
price $r = \phi^{-1}(0)$, known as the \emph{Myerson Reserve Price}.
It follows that the $n$-agent Myerson Revenue can be bounded as follows:
\begin{align*}
\rev^M_n
& = \Expect{\max_i \phi(v_i)^+}
\; \leq \; \Expect{\sum_i \phi(v_i)^+}
\; = \; n \cdot r \cdot (1-F(r)).
\end{align*}

\subsubsection{Posted-Price Mechanisms and the Prophet Inequality}
A natural mechanism for selling a good is the \emph{sequential posted prices}
mechanism.
In round $i$, if the good has not been sold previously,
the mechanism offers the good to agent $i$ at a price of $p_i$.
The revenue obtained by this mechanism is
$\sum_{i=1}^n p_i \cdot
  \Prob{v_i \geq p_i \text{ and } v_j < p_j \text{ for all } j < i}
= \sum_{i=1}^n p_i \cdot (1-F(p_i)) \cdot \prod_{j < i} F(p_j)$.
Because the sequential posted-price mechanism is incentive compatible,
and $\rev^M_n$ is defined as the optimum expected revenue for any
incentive compatible mechanism, we obtain that for any price vector
$(p_1, \hdots, p_n)$:
\begin{align*}
\sum_{i=1}^n p_i \cdot (1-F(p_i)) \cdot \prod_{j < i} F(p_j)
& \leq \rev^M_n.
\end{align*}
The result known as the \emph{Prophet Inequality} guarantees the
existence of a price $p^*$ (called the \emph{prophet price}) such that
a sequential posted-price mechanism with uniform price $p^*$
(i.e., where $p_i = p^*$ for all $i$) generates at least half of the
optimal revenue. In other words:
\begin{align*}
\sum_{i=1}^n p^* \cdot (1-F(p^*)) \cdot F(p^*)^{i-1}
& \geq \half \rev^M_n.
\end{align*}
The Prophet Inequality (and its variants) is a powerful tool in
optimal stopping theory; it was introduced and applied in algorithmic
mechanism design by Hajiaghayi, Kleinberg and Sandholm
\cite{HajiaghayiKS07}. See \cite{alaei11} and \cite{kleinberg12} for
recent developments of the topic.

\section{Pricing globally public goods}\label{sec:global}

In this section, we focus on the case of a \emph{globally} public good.
That is, the underlying network is a clique, $G = K_n$.
We assume that the common value distribution $F$ of the agents is
atomless and regular. Our main result is the following:

\begin{theorem} \label{thm:main-clique}
In the globally public good setting, let
$p = F^{-1}(1-1/n) \cdot (1-1/n)^{n-1}$.
Then, if the price $p$ is offered to all agents, the worst-case
revenue among the equilibria $\Ts \in \nash_{p}$ is at least a
constant fraction of the revenue of the best equilibrium for the best
(possibly non-uniform) price vector to offer the agents.
\end{theorem}

The main insight driving Theorem \ref{thm:main-clique} is that, at equilibrium,
the agents aim to make purchasing decisions in such a way that only one
agent will buy the product, in expectation.  With this in mind, we draw a
relationship between the public good pricing problem and a single item
auction that attempts
to sell a single item to $n$ bidders with value distributions $F$.
We relate the revenue at different price vectors and equilibria in
the public good mechanism to the optimal (Myerson) revenue in
the single item auction.
{We can then apply the theory of optimal
auctions to guide our choice of pricing in the public good mechanism.
We note that similar techniques have been applied in the context of
sequential posted pricing for multi-item auctions \cite{chawla07}.
However, a novel difficulty that we must overcome is the existence of multiple
equilibria of bidder behavior for any given price;
we must therefore find a price for which \emph{all} equilibria generate
a good approximation to the optimal revenue.}

First, in Proposition~\ref{prop:clique-upper-bound}, we show that the
revenue of any equilibrium of any mechanism is upper-bounded by
$\rev^M_n$.
Next, Lemma~\ref{lem:const-approx} shows that for the price
vector $\vc{p} = p \cdot \vc{1}$, where $p$ is the price specified in the assertion of Theorem \ref{thm:main-clique},
the \emph{symmetric} equilibrium is guaranteed to achieve at
least a constant fraction of the Myerson Revenue.
Finally, in Lemma~\ref{lem:any-clique}, we show that in \emph{every}
equilibrium for this price vector $\vc{p}$, the revenue is at least a constant fraction of that
of the symmetric equilibrium for this price vector.

A corollary of this analysis is that the ability to
price-discriminate does not substantially influence revenue: a uniform price
vector can extract a constant fraction of the optimal revenue attainable
by any mechanism, and hence any (non-uniform) vector of prices.

We note that while our analysis makes use of a connection to the Myerson
optimal auction, offering the Myerson Reserve Price does not
necessarily extract a constant fraction of the optimal revenue, even
when $F$ is regular.
In the appendix \ref{appendix:example}, we provide an example illustrating this 
revenue gap at the Myerson Reserve Price.

\begin{proposition} \label{prop:clique-upper-bound}
Let $\vc{p} = (p_1, \ldots, p_n)$ be any price vector,
and $\Ts = (T_1, \ldots, T_n) \in \nash_{\vc{p}}$ be an arbitrary
equilibrium of the public goods selling game with prices $\vc{p}$.
Then, $\rev(\vc{p},\Ts) \leq \rev^M_n$.
\end{proposition}

\begin{proof}
Using that $p_i = T_i \cdot \prod_{j \neq i} F (T_j)$ by
Equation~\eqref{eqn:equilibrium-condition}, we can bound the revenue as
\begin{align*}
\rev(\vc{p},\Ts)
& = \sum_i T_i \cdot (1-F(T_i)) \cdot \prod_{j \neq i} F(T_j)
\; \leq \; \sum_i T_i \cdot (1-F(T_i)) \cdot \prod_{j < i} F(T_j)
\; \leq \; \rev^M_n.
\end{align*}
In the last inequality, we used that the sum expresses the expected
revenue of the sequential posted price mechanism in which the \Kth{i}
player is offered a price of $T_i$; therefore, the sum is
upper-bounded by the expected revenue of the optimal mechanism for
selling a single item.
\end{proof}

For the remainder of this section,
we fix $T$ such that $F(T) = 1-\frac{1}{n}$,
and $p = T \cdot F(T)^{n-1} = T \cdot (1-1/n)^{n-1}$.
Let $\vc{p} = p \cdot \vc{1}$ be the vector in which all agents are
offered $p$.

\begin{lemma} \label{lem:const-approx}
Let $\Ts = T \cdot \vc{1}$ be the symmetric equilibrium corresponding
to $\vc{p}$. Then,
\begin{align*}
\rev(\vc{p},\Ts)
& = n \cdot T \cdot (1-F(T)) \cdot F(T)^{n-1}
\; \geq \; \quarter \cdot \rev^M_n.
\end{align*}
\end{lemma}

\begin{emptyproof}
We use a variant of an argument by Chawla, Hartline and Kleinberg
\cite{chawla07}.
We distinguish between two cases, based on the relation of the Myerson
Reserve Price $r = \phi^{-1}(0)$ with $T$.

\begin{enumerate}
\item If $T > r$, we let $\nu = \phi(T) > 0$.
We can bound the Myerson Revenue as follows:
\begin{multline*}
\rev^M_n
= \Expect{\max_i \phi(v_i) \cdot \one \{ \max \phi (v_i) \geq 0 \}} \\
\leq \;
\nu \cdot \Prob{0 \leq \max_i \phi(v_i) \leq \nu}
+ \Expect{\max \phi(v_i) \cdot \one \{ \max \phi (v_i) \geq \nu \}}.
\end{multline*}
We bound each term separately.
For the first term, we have that
$\nu = \phi(T) \leq T$, and
\begin{align*}
\Prob{0 \leq \max_i \phi(v_i) \leq \nu}
& \leq  \Prob{\max_i v_i \leq T}
\; = \; F(T)^n
\; = \; (1-1/n)^n
\; \leq \; 1/e.
\end{align*}
For the second term, we have that
\begin{align*}
\Expect{\max \phi(v_i) \cdot \one \{ \max \phi (v_i) \geq \nu \}}
& \leq \Expect{\sum_i \phi(v_i) \one\{\phi(v_i) \geq \nu \}}.
\end{align*}
By Lemma~\ref{lem:myerson-lemma},
$\Expect{\phi(v_i) \one\{\phi(v_i) \geq \nu \}}$ is the
revenue of the single-agent mechanism that makes agent $i$ a
take-it-or-leave-it offer at price $T$; therefore,
\begin{align*}
\Expect{\sum_i \phi(v_i) \one\{\phi(v_i) \geq \nu \}}
& = \sum_i T \cdot \Prob{v_i \geq T}
\; = \; T \cdot \sum_i (1-F(T))
\; = \; T,
\end{align*}
by definition of $T$.
Combining the bound on the two terms, we get that
$\rev^M_n \leq T \cdot (1 + 1/e)$.
On the other hand, for the symmetric prices and symmetric equilibrium,
we have that
$\rev(\vc{p},\Ts) = n \cdot T \cdot (1-F(T)) \cdot F(T)^{n-1}
= T \cdot (1-1/n)^{n-1} \geq T/e$.
Therefore,
\begin{align*}
\rev(\vc{p},\vc{T})
& \geq \frac{1/e}{1 + 1/e} \rev^M_n
\; = \; \frac{1}{1+e} \rev^M_n
\; \approx \; 0.27 \cdot \rev^M_n.
\end{align*}

\item When $T \leq r$, we upper-bound $\rev^M_n$ as follows:
\begin{align*}
\rev^M_n
& = \Expect{\max \phi(v_i) \cdot \one \{ \max \phi (v_i) \geq 0 \}}
\; \leq \; \Expect{\sum_i \phi(v_i) \cdot \one \{ \phi (v_i) \geq 0 \}}\\
& =  n \cdot r \cdot (1-F(r)),
\end{align*}
where the final equality follows from the same argument about a single
buyer as above.
Let $q^M = 1 - F(r)$ be the probability that the valuation of an agent
with distribution $F$ is above the Myerson reserve price $r$.
Because $F$ is regular, as argued in Section~\ref{sec:review}, the
revenue curve $R(q)$ is a concave function.
By the definition of the Myerson Reserve Price as the maximizer of
expected revenue, $R$ is maximized at $q = q^M$.
Because we are in the case that $T \leq r$, we get that
$q^M = 1-F(r) \leq 1-F(T) = \frac{1}{n} < 1$.
We can therefore write $\frac{1}{n}$ as a convex combination
$\frac{1}{n} = \lambda \cdot q^M + (1-\lambda) \cdot 1$,
with $\lambda=\frac{1-\frac{1}{n}}{1-q^M}$.
The concavity of $R$, together with $R(1) = 0$, now implies that
$R(\frac{1}{n}) \geq \frac{1-\frac{1}{n}}{1-q^M} R(q^M)$.
On the other hand, $R(\frac{1}{n}) = T \cdot (1-F(T))$;
by combining these, we obtain
\begin{align}
\label{eq:eq1}
T \cdot (1-F(T))
& = R(1/n)
\; \geq \; \frac{1-\frac{1}{n}}{1-q^M} \cdot R(q^M)
\; \geq \; \left(1-\frac{1}{n}\right) \cdot r \cdot  (1-F(r)).
\end{align}
We can therefore bound the posted price revenue as
\begin{align*}
\rev(\vc{p},\Ts)
& = n \cdot T \cdot (1-F(T)) \cdot F(T)^{n-1}
\; = \; n \cdot T \cdot (1-F(T)) \cdot (1-1/n)^{n-1} \\
& \geq \; n \cdot r \cdot (1-F(r)) (1-1/n)^n
\; \geq \; \quarter \rev^M_n.
\end{align*}
The first inequality follows by Equation~\eqref{eq:eq1};
for the second inequality, we bound $(1-1/n)^n \geq \quarter$, and
use that the optimal revenue from selling a single item to $n$ agents
is at most $n$ times the optimal revenue from selling a single item to
one agent at the Myerson Reserve Price. \QED
\end{enumerate}
\end{emptyproof}

Having shown that the symmetric equilibrium has revenue within a
constant factor of the Myerson Revenue for a single item, it remains
to analyze the asymmetric equilibria.
(Recall that $p = T \cdot F(T)^{n-1}$, where $T$ is such that $F(T) = 1-\frac{1}{n}$.)

\begin{lemma} \label{lem:any-clique}
Let $\Ts = T \cdot \vc{1}$ be the symmetric equilibrium with threshold
$T$, and $\TP \in \nash_{\vc{p}}$ be an arbitrary equilibrium.
Then, $\rev(\vc{p}, \TP) \geq \Omega(1) \cdot \rev(\vc{p}, \Ts)$.
\end{lemma}

\begin{proof}
We express the revenue of the symmetric equilibrium as
$\rev(\vc{p}, \Ts) = p \cdot \sum_i (1-F(T)) = p$.
By the Union Bound, the revenue at the equilbrium \TP is lower-bounded
by
$\rev(\vc{p}, \TP) = \sum_i p \cdot (1-F(T'_i))
\geq p \cdot (1 - \prod_i F(T'_i))$.
We will prove that $\prod_i F(T'_i) \leq (1-1/n)^{n-1} \leq \half$,
which will imply that
$\rev(\vc{p}, \TP) \geq \half p = \half \rev(\vc{p}, \Ts)$.
For contradiction, assume that $\prod_i F(T'_i) > (1-1/n)^{n-1}$.

Using that $p = T \cdot F(T)^{n-1} = T \cdot (1-1/n)^{n-1}$,
applying the equilibrium condition
\eqref{eqn:equilibrium-condition} to $\TP$,
and using our contradiction assumption,
we get that for all $j$,
\begin{align*}
T'_j
& = \frac{p}{\prod_{i \neq j} F(T'_i)}
\; = \; \frac{T \cdot (1-1/n)^{n-1}}{\prod_{i \neq j} F(T'_i)}
\; \leq \; T \cdot \frac{(1-1/n)^{n-1}}{\prod_{i} F(T'_i)}
\; < \; T \cdot \frac{(1-1/n)^{n-1}}{(1-1/n)^{n-1}}
\; = \; T.
\end{align*}
Thus, $T'_j < T$ for all $j$, implying that $F(T'_j) \leq F(T)$ as well.
But this contradicts that
$T'_j \cdot \prod_{i \neq j} F(T'_i)
= p = T \cdot \prod_{i \neq j} F(T)$, completing the proof.
\end{proof}

\section{Pricing Locally Public Goods}
\label{sec:simultaneous}

In this section, we turn to scenarios in which the good is not
completely public. That is, the graph $G$ is not
necessarily complete; rather, $G$ is an arbitrary network, and 
agents share benefits only with neighbors in $G$.
We refer to such a good as \emph{locally public}.

We first analyze the case when $G$ is $d$-regular, for some arbitrary
$d \geq 1$.  For such graphs, we describe how to explicitly calculate
prices that are approximately revenue-optimal, in the worst case over
equilibria of agent behavior.  
We then consider the case of general networks.
We present evidence that the pricing problem for general graphs is
substantially more difficult, and that the approaches used in previous
cases cannot be extended to handle the general case, 
even for uniform distributions.
Nevertheless, we show that approximately optimal prices
can be found in the special case that agent values are drawn from the uniform
distribution.

\subsection{$d$-Regular Graphs}

We consider the problem of pricing locally public goods
when the underlying graph is $d$-regular;
i.e., $\abs{N(i)} = d$ for every $i \in [n]$.
As before, we assume that the value distribution $F$ of the agents is
atomless and regular.

Recall that in the case of globally public goods
(Section~\ref{sec:global}), we showed how to compute a price for which
the revenue at the \emph{worst} equilibrium is a good approximation to the revenue
at the \emph{best} equilibrium for any price vector.
In other words, $p$ is such that
$$
 \min_{\Ts \in
\nash_{p}}
\rev(p \cdot \vc{1}, \Ts) \geq \Omega(1) \cdot \max_{\vc{p}} \max_{\Ts \in
\nash_{\vc{p}}} \rev(\vc{p},\Ts).
$$
One might hope for a similar result for locally public goods. 
Unfortunately, we show that this is not possible even for $d$-regular networks:
in Example \ref{ex:d-reg-revenue}, we give an instance of a
$d$-regular graph for which the gap in revenue between different
equilibria is linear in $d$.
The same example also shows that for $d$-regular graphs, we cannot find
a single price that is competitive against non-uniform price vectors.
Thus, unlike for globally public goods, a constant-factor 
revenue approximation for $d$-regular graphs must specifically compare
revenue-minimizing equilibria at given price vectors.

We establish the existence of a price $p$ that depends only on
the degree $d$ and the distribution $F$, but not on the particular
structure of $G$, such that when $p$ is offered to all the agents, 
the seller obtains a constant fraction of the \emph{worst-case}
revenue at \emph{any} price.
In other words, we establish the existence of a price $p = p(d,F)$
such that
$$\min_{\Ts \in
\nash_{p}}
\rev(p \cdot \vc{1}, \Ts) \geq \Omega(1) \cdot \max_{p'} \min_{\Ts \in
\nash_{p'}}
\rev(p' \cdot \vc{1},\Ts).\\$$



{We emphasize that the key aspect here is that $p$ is
independent of the actual network structure of $G$, and that it can be
computed efficiently from $F$ and $d$.}

\begin{theorem}
\label{thm:main-d-regular}
In the locally public good setting with $d$-regular graphs,
let $p = F^{-1}(1-1/d) \cdot (1-1/d)^{d}$.
Then, if the price $p$ is offered to all agents, the worst-case
revenue among the equilibria $\Ts \in \nash_{p}$ is at least a
constant fraction of the revenue of the worst equilibrium for the best
\emph{network-specific} uniform price to offer the agents.
\end{theorem}

Our approach to proving Theorem \ref{thm:main-d-regular} is the following.
We first study the symmetric equilibria of the game.
We know from Section \ref{sec:models} that every uniform price vector
$p \cdot \vc{1}$ admits a symmetric equilibrium.
We consider a price $p$ for which, in the corresponding symmetric equilibrium,
each player buys with probability $\frac{1}{d}$.  
In Lemma \ref{lemma:regular_symmetric_approx}, we
show that the revenue of this symmetric equilibrium is a constant
fraction of the revenue of any other symmetric equilibrium (across all
potential prices). 
In particular, this implies that the worst-case revenue at any other
price is at most a constant factor larger than the revenue of the
symmetric equilibrium at price $p$.
Then, in Lemma \ref{lemma:regular_symmetric_assymmetric}, we show that for this
particular price $p$, \emph{every} equilibrium generates at least a
constant fraction of the revenue of the symmetric equilibrium.

For the remainder of this section,
we fix $T$ such that $F(T) = 1-\frac{1}{d}$,
and $p = T \cdot F(T)^d = T \cdot (1-1/d)^{d}$.
Let $\vc{p} = p \cdot \vc{1}$ be the vector in which all agents are
offered $p$.

\begin{lemma}\label{lemma:regular_symmetric_approx}
Consider a locally public goods problem in which the underlying network is a
$d$-regular graph and agents have valuations drawn i.i.d.~from a regular
distribution $F$. 
Let $\vc{T}$ be the symmetric equilibrium with threshold $T$. Then, 
$$\rev(\vc{p}, \vc{T}) \geq \Omega(1)
\cdot \rev(p'\cdot \vc{1}, T' \cdot \vc{1}),$$  
for all prices $p'$ and threshold vectors $T' \cdot \vc{1}$
corresponding to the symmetric equilibrium with price $p'$.
\end{lemma}

\begin{proof}
Let $\rev^M_d$ be the revenue obtained by Myerson's mechanism for
selling one item to $d$ players with i.i.d.~valuations drawn according
to $F$. The Prophet Inequality (Section~\ref{sec:review}) guarantees
that there exists a price $T^*$ such that a 
sequential posted-prices mechanism with price $T^*$ offered to all agents 
guarantees at least half of the Myerson Revenue $\rev^M_d$.
On the other hand, the Myerson Revenue is optimal, and therefore
clearly serves as an upper bound on the revenue that can be obtained
by any price $T$ of the sequential posted prices mechanism.
In summary, there exists a $T^*$ such that
\begin{equation}
\sum_{i=1}^d T^* (1-F(T^*)) F(T^*)^{i-1} \geq \half \cdot \rev^M_d
\geq \half \cdot \Big( \sum_{i=1}^d T' (1-F(T')) F(T')^{i-1} \Big),
\quad \text{ for all } T'.
\label{eq:prophet}
\end{equation}

We distinguish between two cases, based on the relation of $T^*$ with $T$.

\begin{enumerate}
\item If $T > T^*$, then given any price $p'$ and corresponding
symmetric equilibrium $T'$,
\begin{align*}
\rev(p'\cdot \vc{1}, T' \cdot \vc{1}) 
& = n \cdot T' \cdot (1-F(T')) \cdot F(T')^d 
\; = \; \frac{n}{d} \left(d \cdot T' \cdot (1-F(T')) \cdot F(T')^d \right) \\
& \leq \frac{n}{d} 
\Big( \sum_{i=1}^d T' \cdot (1-F(T')) \cdot F(T')^{i-1} \Big).
\end{align*}
Using both sides of Equation~\eqref{eq:prophet}, we get: 
\begin{align}
\rev(p'\cdot \vc{1}, T' \cdot \vc{1})  
& \leq \frac{n}{d} \cdot \rev^M_d 
\; \leq \; \frac{2n}{d} \cdot
\Big( \sum_{i=1}^d T^* (1-F(T^*)) F(T^*)^{i-1} \Big) 
\; \leq \; \frac{2 n}{d}\cdot T^*.
\label{eqn:comparison-star}
\end{align}

Next, we establish a lower bound on $\rev(\vc{p}, \vc{T})$.
\[
\rev(\vc{p}, \vc{T}) 
=  n \cdot T \cdot (1-F(T)) \cdot F(T)^d 
= \frac{n}{d} (1-1/d)^d \cdot T \geq \frac{n}{4d} \cdot T^* 
\geq \frac{1}{8} \cdot \rev(p'\cdot \vc{1}, T' \cdot \vc{1}).
\]
For the first inequality, we bound  $\left( 1- \frac{1}{d} \right)^d \geq
\frac{1}{4}$, and use that $T \geq T^*$, by the assumption of case (1); the
second inequality follows from Inequality \eqref{eqn:comparison-star}.

\item 
\noindent If $T \leq T^*$, then let $q^* = 1 - F(T^*) \leq 1
- F(T) = \frac{1}{d}$. Similar to the proof of
Lemma~\ref{lem:const-approx},
we use the concavity of the revenue function 
$R(q) = q \cdot F^{-1}(1-q)$ to derive that
\begin{equation}
T \cdot (1-F(T)) \geq (1-1/d) \cdot T^* \cdot (1-F(T^*)).
\label{eq:concavity}
\end{equation}
It follows that
\begin{align*}
\rev(\vc{p}, \vc{T}) 
& =  n \cdot T \cdot (1-F(T)) \cdot F(T)^d  
\; = \; n \cdot (1-1/d)^d \cdot T \cdot (1-F(T)) \\
& \geq n (1-1/d)^{d+1} \cdot T^* \cdot (1-F(T^*)) \\
& \geq (1-1/d)^{d+1} \cdot \frac{n}{d} \cdot
\Big( \sum_{i=1}^d T^* (1-F(T^*)) F(T^*)^{i-1} \Big) \\ 
& \geq \frac{n}{16 d} \cdot 
\Big(\sum_{i=1}^d T' (1-F(T')) F(T')^{i-1} \Big) \\
& \geq \frac{n}{16} \cdot T' \cdot (1-F(T')) \cdot F(T')^d 
\; = \; \frac{1}{16} \cdot \rev(p'\cdot \vc{1}, T' \cdot \vc{1}).
\end{align*}
The first inequality follows by Inequality~\eqref{eq:concavity};
the second inequality holds because $ F(T^*)^{i-1} \leq 1$ for every $i$;
the third inequality follows by Inequality~\eqref{eq:prophet}, 
and by bounding $(1-1/d)^{d+1} \geq \frac{1}{8}$;
and for the final inequality, we use that $i \leq d$ for every $i$.
\end{enumerate}
The assertion of the lemma follows.
\end{proof}

Having shown that the symmetric equilibrium associated with price $p$ generates
a  good approximation to the optimal revenue  attainable at symmetric
equilibria, we now show that there are no other (asymmetric)  equilibria
associated with price $p$ that generate significantly less revenue.

\begin{lemma}\label{lemma:regular_symmetric_assymmetric}
Let $\vc{T} = T \cdot \vc{1}$ be the symmetric threshold vector
associated with price $p$. 
Then, $\rev(\vc{p}, \vc{T}') \geq \Omega(1) \cdot \rev(\vc{p}, \vc{T})$
for any threshold vector $\vc{T'} \in \nash_{\vc{p}}$.
\end{lemma}

\begin{emptyproof}
For any equilibrium $\vc{T}' \in \nash_{\vc{p}}$,
$$\rev(\vc{p}, \vc{T}') = \sum_i p (1-F(T'_i)) = \frac{p}{d} \sum_i \sum_{j
\in N(i)}
(1-F(T'_j)) \geq \frac{p}{d} \sum_i \Big( 1 - \prod_{j \in N(i)} F(T'_j)
\Big),
$$
where the last inequality follows by applying the union bound.
By the equilibrium conditions, 
$T'_i \cdot \prod_{j \in N(i)} F(T'_j) = p = T \cdot F(T)^d$
for all $i$; therefore, 
$1 - \prod_{j \in N(i)} F(T'_j) = 1- T \cdot F(T)^d / T'_i$.
We get that
$
\rev(\vc{p}, \vc{T}') \geq \frac{p}{d} \sum_i \left( 1- \frac{T \cdot F(T)^d}{T'_i} \right).
$
From the last inequality and the equality 
$\rev(\vc{p}, \vc{T}') = \sum_i p (1-F(T'_i))$, we can bound 
$\rev(\vc{p}, \vc{T}')$ as follows: 
\[
\rev(\vc{p}, \vc{T}') \geq 
p \cdot \sum_i \half \left[ ( 1- F(T'_i)) 
  + \frac{1}{d} \left( 1-  \frac{T \cdot F(T)^d}{T'_i}\right)
\right]. 
\]
Focus on one term $i$ of the sum.
The first term in brackets is decreasing in $T'_i$, while the second
is increasing in $T'_i$.
Consequently, we distinguish between two cases: 
(i) If $T'_i \leq T$, then $1- F(T'_i) \geq 1-F(T) = \frac{1}{d}$. 
(ii) If $T'_i \geq T$, then
\[
\frac{1}{d} \left(1-  \frac{T \cdot F(T)^d}{T'_i} \right)
\geq \frac{1}{d} \left( 1 - (1-1/d)^d \right) 
\geq \frac{1}{d} \cdot \left( 1-\frac{1}{e} \right).
\]
Thus, summing over all $i$, 
$\rev(\vc{p}, \vc{T}') \geq p \cdot \sum_i \half \left( \frac{1}{d}
\cdot \left( 1-\frac{1}{e}\right)  \right)$,
implying that
$$\rev(\vc{p}, \vc{T}') 
\geq p \cdot \sum_i \half \left( \frac{1}{d}
\cdot (1-1/e)  \right) = \Omega(1) \cdot p \cdot
\frac{n}{d} = \Omega(1) \cdot \rev(\vc{p}, \vc{T}). \QED$$
\end{emptyproof}

We now show that comparing against the best worst-case revenue,
rather than the best revenue in equilibrium, is a necessity rather than
an artifact of our analysis.

\begin{example}[Revenue gap]
\label{ex:d-reg-revenue}
Consider an instance with $n$ players whose valuations are drawn
i.i.d.~from the uniform distribution with support $[0,1]$.
Let the underlying network be a $d$-regular bipartite
graph with $\frac{n}{2}$ nodes on each side. 
We showed in Lemma~\ref{lemma:regular_symmetric_approx} that the best
worst-case revenue is upper bounded by 
$\frac{n}{d} \cdot \rev^M_d \leq \frac{n}{d}$
(where $\rev^M_d$ is the revenue obtained by Myerson's mechanism for
selling one item to $d$ players).

Now, consider the following equilibrium $\vc{T} \in \nash_{1/2}$ in the
bipartite graph: the nodes on one side buy whenever their
value exceeds the price, while the nodes on the other side
always free-ride.
That is, $T_i = \half$ for each player $i$ on the left, and 
$T_i = 2^{d-1} > 1$ for each player $i$ on the right.
This equilibrium generates a revenue of 
$\frac{n}{2} \cdot \half \cdot (1-\half) = \frac{n}{8}$.
The gap between the best worst-case revenue and the best revenue can
therefore be as large as 
$\frac{n}{8} \cdot (\frac{n}{d})^{-1} = \frac{d}{8}$.

Notice that the same instance also shows a gap between the worst-case
revenue of the best uniform price vector and the best discriminating
prices. The seller can offer all nodes on the right a price
of 1 and the left a price of \half; in the unique equilibrium, the bidders
on the right never buy and the bidders on the left choose a threshold
of \half.
\end{example}

\subsection{Hardness of Bounding Revenue for General Graphs}
\label{sec:simult-hard}

We would like to extend the results from the previous sections
beyond complete and $d$-regular
graphs, and find a method to compute prices that
approximately optimize revenue for arbitrary networks.
Recall the nature of our analysis for Theorem~\ref{thm:main-clique} and
Theorem~\ref{thm:main-d-regular}: in each case, we constructed a price
$p$ and then bounded the revenue of the worst-case equilibrium 
$\Ts \in \nash_{p}$ with respect to either an upper bound on the
revenue of \emph{any} equilibrium for \emph{any} price vector 
(in the case of Theorem~\ref{thm:main-clique}) or the worst-case
revenue for any uniform price vector 
(for Theorem~\ref{thm:main-d-regular}).
Can we hope to extend these methods to general networks?

In this section, we show that there are inherent difficulties in
extending these approaches to handle general networks.
In appendix \ref{appendix:example}, we give an instance of a
network such that, for \emph{every} price, the gap between the
best-case and worst-case revenues is $\Omega(n)$.
(The complete bipartite graph $K_{n/2,n/2}$, generalizing 
Example~\ref{ex:d-reg-revenue}, shows the same for carefully chosen
prices, but not all prices.)

One might instead hope to analyze worst-case revenues directly, as in
Theorem~\ref{thm:main-d-regular}.
However, we again find that this poses a difficulty in general networks.
Theorem \ref{thm:NP-hard-rev} (whose proof is given in appendix \ref{app:NP-hard})
shows that it is NP-hard to approximate the worst-case revenue for a
given $p$, over all equilibria $\Ts \in \nash_{p}$, to within a factor
of $n^{1-\epsilon}$, even when $F$ is the uniform distribution.

\begin{theorem}
\label{thm:NP-hard-rev}
Assume that all agents' valuations are drawn i.i.d.~from the uniform
distribution on $[0,1]$.
Given a graph $G$ with $n$ nodes and a uniform price vector 
$\vc{p} = p\cdot \vc{1}$, it is NP-hard to approximate the worst-case
revenue to within a factor $n^{1-\epsilon}$.
\end{theorem}

\subsection{General Graphs, Uniform Distribution}
\label{sec:simult-ND}
Motivated by the gap between best-case and worst-case revenue, and the
approximation hardness, we explore an alternative approach.
While Theorem~\ref{thm:NP-hard-rev} demonstrates that we cannot
hope to compute the revenue generated by any given price,
we show that for i.i.d.~uniform distributions of agent valuations, the
impact of the equilibrium and of the price choice can be decoupled, so
that an approximately optimum price can be set even without knowledge
of the network. It turns out that a price of \half gives a
constant-factor approximation.

The key to our analysis is to show that, for the case of the uniform
distribution, there is an underlying structure to each equilibrium
that does not depend on the price chosen by the seller.
Even further, it is possible to express revenue as the product of two
terms, the first determined by the chosen price and the second 
by the structure of the equilibrium selected by the agents.
This allows us to optimize the price term independently of the
equilibrium structure.

\begin{theorem}\label{thm:uniform_general_network}
Let $G$ be an arbitrary network, and assume that the agents'
valuations are drawn i.i.d.~from the uniform distribution on $[0,1]$.
Then, the worst-case revenue obtained from offering a uniform price of
$\frac{1}{2}$ is at least an $\frac{e}{4}$ fraction of the
worst-case revenue for the optimum (network-specific) price. 
Formally: 
\[
\min_{\Ts \in \nash_{1/2}} \rev({\textstyle\frac{1}{2}} \cdot \vc{1}, \Ts) 
\; \geq \; \frac{e}{4} \cdot \max_{\vc{p} = p \cdot \vc{1}} 
\min_{\Ts \in \nash_{\vc{p}}} \rev(\vc{p},\Ts).
\]
\end{theorem}
\begin{emptyproof}
%
Given a price vector $\vc{p} = p \cdot \vc{1}$, an equilibrium 
$\vc{T} \in \nash_{\vc{p}}$ is a vector such that 
$T_i \cdot \prod_{j \in N(i)} F(T_j) = p$, 
where $F(T_j) = \min \{1, T_j\}$ for the uniform distribution on $[0,1]$.
Note that a threshold $T_i > 1$ is ``behaviorally'' equivalent to a
threshold $T_i = 1$, since the support of the valuations is $[0,1]$.
Applying this definition of the distribution function, the 
equilibrium condition becomes
\begin{align*}
\text{(i) } T_i & \in [p,1]; 
& \text{(ii) } \prod_{j \in N(i) \cup \{i\}} T_j & \leq p; 
& \text{(iii) } \prod_{j \in N(i) \cup \{i\}} T_j < p \implies T_i = 1.
\end{align*}

The worst-case equilibrium for the price vector $\vc{p}$ can therefore
be expressed as the following mathematical program:

\begin{LP}{$\text{Minimize}_{\Ts \in \nash_{p}}$}{%
\rev(p \cdot \vc{1},\Ts) \; = \; p \cdot \textstyle\sum_i (1-T_i)}
\prod_{j \in N(i) \cup \{i\}} T_j \leq p & \text{for all } i\\
\prod_{j \in N(i) \cup \{i\}} T_j < p \implies T_i = 1 
& \text{for all } i \\
p \leq T_i \leq 1 & \text{for all } i.
\end{LP}


We use the transformation
$x_i = \frac{\log(1/T_i)}{\log(1/p)}$ (and thus $T_i = p^{x_i}$)
in order to bring the program into a form in which the constraints
carry only information about the graph and are independent of the
price $p$.
In addition, as a result, the objective function depends only on the
price and not on the graph structure. 
This decouples the two aspects of the problem:

\begin{LP}[eqn:LP]{$\text{Minimize}_{\Ts \in \nash_{p}}$}{%
\rev(p \cdot \vc{1}, \Ts) 
\; = \; p \cdot \sum_i (1-\exp( -x_i \cdot \log(1/p)))}
\sum_{j \in N(i) \cup \{i\}} x_j \geq 1 & \text{for all } i \\
\sum_{j \in N(i) \cup \{i\}} x_j > 1 \implies x_i = 0
& \text{for all } i \\
0 \leq x_i \leq 1 & \text{for all } i.
\end{LP}

For the range $y \in [0, \log(1/p)]$, elementary facts about the
exponential function imply the following bounds:
$\frac{(1-p)}{\log(1/p)} \cdot y 
 \leq 1-e^{-y} 
\; \leq \; y$. 
Writing $X$ for the set of vectors $\vc{x}$ that are feasible for
the program~\eqref{eqn:LP}, we apply the bound on the exponential
function to the program~\eqref{eqn:LP}, obtaining that
\begin{align*}
p \cdot (1-p) \cdot \min_{\vc{x} \in X} \sum_i x_i  
& \leq \min_{\Ts \in \nash_{p}} \rev(p \cdot \vc{1}, \Ts)
\; \leq \;  p \cdot \log(1/p) \cdot \min_{\vc{x} \in X} \sum_i x_i.
\end{align*}

Thus, we have upper and lower bounds on the value of the worst-case
revenue for each price $p$.
Notice that the upper bound is maximized for $p = 1/e$, so
\begin{align*}
\max_{\vc{p} = p \cdot \vc{1}} \min_{\Ts \in \nash_{\vc{p}}} \rev(\vc{p},\Ts)
& \leq \frac{1}{e} \cdot \min_{\vc{x} \in X} \sum_i x_i.
\end{align*}
On the other hand, setting $p = \frac{1}{2}$ maximizes the
lower-bound, giving us that
\begin{align*}
\min_{\Ts \in \nash_{1/2}} \rev(\frac{1}{2} \cdot \vc{1}, \Ts) 
& \geq \frac{1}{4} \cdot \min_{\vc{x} \in X} \sum_i x_i 
\; \geq \; \frac{e}{4} \cdot \max_{\vc{p} = p \cdot \vc{1}} 
  \min_{\Ts \in \nash_{\vc{p}}} \rev(\vc{p},\Ts). \QED
\end{align*}
\end{emptyproof}

The analysis above was based on choosing the Myerson Price for the
uniform distribution, in order to maximize the ``price component'' of
the product; the ``equilibrium component'' factored out, and
contributed at most a constant-factor loss in revenue.
One might suspect that the Myerson Price would provide a constant
approximation for all (regular) distributions.
However, in appendix \ref{appendix:example}, we provide an example
which shows that this is not the case, even for the complete network.

\section{Concluding Remarks} \label{sec:concl}

In this paper, we initiated the study of revenue-maximal pricing
for locally public goods. We conclude by discussing
potential extensions and questions left open by our work.


\paragraph{Simultaneous vs.~Sequential purchases} In our model, all agents
simultaneously observe the price of the good and decide whether
to purchase. In an alternative scenario,
the seller sequentially offers the good to each agent in turn, who can then
decide whether to buy given the choices of those who came before.
%
This leads to a pricing problem that is similar to the one studied
here, except that the natural solution concept for the pricing game
becomes the subgame perfect equilibrium rather then Nash Equilibrium.
Does the increased possibility of coordination in sequential sales
unambiguously help or harm revenue?

\paragraph{Imperfect public goods} We consider scenarios where the good
is a perfect public good: the benefit of owning it and having a
neighbor that owns it are the same. Most goods, however, have both public and
private components. The purchase of a big-screen television 
provides some benefit to the purchaser's friends,
who can visit and watch/play, but the greatest benefit goes to the purchaser
himself.
%
One can consider a utility model where $u_i(v_i, S) = v_i - p_i$ if $i \in S$,
$(1-\epsilon) v_i$ if $i \notin S$ but $N(i) \cap S \neq \emptyset$,
and $u_i(v_i, S) = 0$ otherwise.

\paragraph{Strength of social ties} We assumed that all links of the social
network are homogeneous. One could also consider the case in which
network links are weighted. 
The weights might correspond to the extent to which benefits
are shared along a link. In such a case,
one can assume that the network is represented by a matrix $w_{ij}$ where
$w_{ii} = 1$, $w_{ij} = 0$ for $j \notin N(i) \cup \{i\}$ and $w_{ij} \in
[0,1]$ otherwise. Then, the utility of agent $i$ for an outcome $S$ is given by
$u_i(v_i,S) = v_i \cdot \max_{j \in S} w_{ij} - p_i \cdot \one\{i \in S\}$. An
even further generalization is to consider a submodular function $f_i :
2^{[n]} \rightarrow \R_+$ for each agent such that his utility is
given by $u_i = v_i \cdot f_i(S) - p_i \cdot \one\{i \in S\}$.

\paragraph{Other applications and objectives} We believe that our model
and techniques can be useful 
for additional related settings.
Consider, for example, the following {\em snow-shoveling} setting.
Suppose that a landlord of an apartment building wants to make sure
that snow is shoveled from the sidewalk in front of his building.
Thereto, he imposes a fine on each tenant in the case that the
sidewalk is not shoveled. The tenants now face a problem that is
similar in spirit to purchasing a public good.
Each tenant incurs a personal cost from snow shoveling, drawn
from some distribution, and needs to decide whether to exert effort
(and incur the associated cost), 
or else pay the fine if none of the other tenants shoveled.
The landlord, in determining the fine, must balance between
different objectives, such as getting the snow
shoveled, his own revenue, and the social welfare of the tenants. 
This is an example of a broader class of problems in which a
policy maker must
decide on mechanisms that are only applicable to individuals
in order to encourage group behaviors. This example illustrates the
appeal of this problem with objectives other then revenue. We believe
that this problem has a structure similar to public goods, and that our
techniques might be useful there.

\bibliographystyle{abbrv}
\bibliography{names,conferences,bibliography,sigproc}

\begin{thebibliography}{10}

\bibitem{akhlaghpour10}
H.~Akhlaghpour, M.~Ghodsi, N.~Haghpanah, H.~Mahini, V.~Mirrokni, and A.~Nikzad.
\newblock Optimal iterative pricing over social networks.
\newblock In {\em Proc. 6th Workshop on Internet and Network Economics (WINE)},
  2010.

\bibitem{alaei11}
S.~Alaei.
\newblock Bayesian combinatorial auctions: Expanding single buyer mechanisms to
  many buyers.
\newblock In {\em Proc. 52nd IEEE Symp. on Foundations of Computer Science},
  pages 512--521, 2011.

\bibitem{Allouch}
N.~Allouch.
\newblock On the private provision of public goods on networks.
\newblock Working Paper 2012.40, Fondazione Eni Enrico Mattei, May 2012.

\bibitem{anari10}
N.~Anari, S.~Ehsani, M.~Ghodsi, N.~Haghpanah, N.~Immorlica, H.~Mahini, and
  V.~Mirrokni.
\newblock Equilibrium pricing with positive externalities.
\newblock In {\em Proc. 6th Workshop on Internet and Network Economics (WINE)},
  2010.

\bibitem{arthur:motwani:sharma:xu}
D.~Arthur, R.~Motwani, A.~Sharma, and Y.~Xu.
\newblock Pricing strategies for viral marketing on social networks.
\newblock In {\em Proc. 5th Workshop on Internet and Network Economics (WINE)},
  pages 101--112, 2009.

\bibitem{Bergstrom86}
T.~Bergstrom, L.~Blume, and H.~R. Varian.
\newblock {On the Private Provision of Public Goods}.
\newblock {\em Journal of Public Economics}, 29(1):25--49, Jan. 1986.

\bibitem{Bhalgat12}
A.~Bhalgat, S.~Gollapudi, and K.~Munagala.
\newblock Mechanisms and allocations with positive network externalities.
\newblock In {\em Proc. 14th ACM Conf. on Electronic Commerce}, pages 179--196,
  2012.

\bibitem{BramoulleKranton07}
Y.~Bramoull\'{e} and R.~Kranton.
\newblock Public goods in networks.
\newblock {\em Journal of Economic Theory}, 135(1):478--494, July 2007.

\bibitem{BramoulleKranton10}
Y.~Bramoull\'{e}, R.~Kranton, and M.~D'Amours.
\newblock Strategic interaction and networks.
\newblock Cahiers de recherche 1018, CIRPEE, 2010.

\bibitem{brocas:endogenous}
I.~Brocas.
\newblock Endogenous entry in auctions with negative externalities.
\newblock {\em Theory and Decision}, 54(2):125--149, 2003.

\bibitem{candogan}
O.~Candogan, K.~Bimpikis, and A.~E. Ozdaglar.
\newblock Optimal pricing in networks with externalities.
\newblock {\em Operations Research}, 60(4):883--905, 2012.

\bibitem{chawla07}
S.~Chawla, J.~D. Hartline, and R.~D. Kleinberg.
\newblock Algorithmic pricing via virtual valuations.
\newblock In {\em Proc. 9th ACM Conf. on Electronic Commerce}, pages 243--251,
  2007.

\bibitem{chawla10}
S.~Chawla, J.~D. Hartline, D.~L. Malec, and B.~Sivan.
\newblock Multi-parameter mechanism design and sequential posted pricing.
\newblock In {\em Proc. 41st ACM Symp. on Theory of Computing}, pages 311--320,
  2010.

\bibitem{cornes:sandler:externalities}
R.~Cornes and T.~Sandler.
\newblock {\em The Theory of Externalities, Public Goods, and Club Goods}.
\newblock Cambridge University Press, 1996.

\bibitem{Galeotti08}
A.~Galeotti, S.~Goyal, M.~O. Jackson, F.~Vega-Redondo, and L.~Yariv.
\newblock Network games.
\newblock Economics Working Paper ECO2008/07, European University Institute,
  2008.

\bibitem{Haghpanah11}
N.~Haghpanah, N.~Immorlica, V.~S. Mirrokni, and K.~Munagala.
\newblock Optimal auctions with positive network externalities.
\newblock In {\em Proc. 13th ACM Conf. on Electronic Commerce}, pages 11--20,
  2011.

\bibitem{HajiaghayiKS07}
M.~T. Hajiaghayi, R.~D. Kleinberg, and T.~Sandholm.
\newblock Automated online mechanism design and prophet inequalities.
\newblock In {\em Proc. 22nd AAAI Conf. on Artificial Intelligence}, pages
  58--65, 2007.

\bibitem{hartline_lectures12}
J.~D. Hartline.
\newblock Approximation in economic design.
\newblock Lecture Notes,
  \url{http://users.eecs.northwestern.edu/~hartline/amd.pdf}, 2012.

\bibitem{hartline:mirrokni:sundarajan}
J.~D. Hartline, V.~S. Mirrokni, and M.~Sundararajan.
\newblock Optimal marketing strategies over social networks.
\newblock In {\em 17th Intl. World Wide Web Conference}, pages 189--198, 2008.

\bibitem{jackson:nicolo}
M.~O. Jackson and A.~Nicol\'o.
\newblock The strategy-proof provision of public goods under congestion and
  crowding preferences.
\newblock {\em Journal of Economic Theory}, 115(2):278--308, 2004.

\bibitem{jehiel:moldovanu:interdependent}
P.~Jehiel and B.~Moldovanu.
\newblock Efficient design with interdependent valuations.
\newblock {\em Econometrica}, 69(5):1237--59, 2001.

\bibitem{jehiel:moldovanu:externalities}
P.~Jehiel and B.~Moldovanu.
\newblock Allocative and informational externalities in auctions and related
  mechanisms.
\newblock In {\em Proc. 9th World Congress of the Econometric Society}, 2006.

\bibitem{jehiel:moldovanu:stacchetti:nuclear}
P.~Jehiel, B.~Moldovanu, and E.~Stacchetti.
\newblock How (not) to sell nuclear weapons.
\newblock {\em American Economic Review}, 86(4):814--829, 1996.

\bibitem{jehiel:moldovanu:stacchetti:multidimensional}
P.~Jehiel, B.~Moldovanu, and E.~Stacchetti.
\newblock Multidimensional mechanism design for auctions with externalities.
\newblock {\em J. Econ. Theory}, 85(2):258--294, 1999.

\bibitem{kleinberg12}
R.~Kleinberg and S.~M. Weinberg.
\newblock Matroid prophet inequalities.
\newblock In {\em Proc. 43rd ACM Symp. on Theory of Computing}, pages 123--136,
  2012.

\bibitem{krishna2009auction}
V.~Krishna.
\newblock {\em Auction Theory}.
\newblock Elsevier Science, 2009.

\bibitem{mas-collel:whinston:green}
A.~Mas-Collel, M.~D. Whinston, and J.~R. Green.
\newblock {\em Microeconomic Theory}.
\newblock Oxford University Press, 1995.

\bibitem{google_internet_chelsea}
P.~McGeehan.
\newblock Free outdoor wi-fi comes to west chelsea.
\newblock NY Times, January 8, 2013, \url{
  http://cityroom.blogs.nytimes.com/2013/01/08/free-outdoor-wi-fi-comes-to-chelsea
  /}, 2012.

\bibitem{milgrom2004putting}
P.~Milgrom.
\newblock {\em Putting Auction Theory to Work}.
\newblock Churchill Lectures in Economics. Cambridge University Press, 2004.

\bibitem{myerson-81}
R.~Myerson.
\newblock Optimal auction design.
\newblock {\em Mathematics of Operations Research}, 6(1):58--73, 1981.

\bibitem{pigou}
A.~Pigou.
\newblock {\em The Economics of Welfare}.
\newblock Macmillan, 1920.

\bibitem{samuelson54}
P.~A. Samuelson.
\newblock The pure theory of public expenditure.
\newblock {\em Review of Economics and Statistics}, 36:387--389, 1954.

\end{thebibliography}

\appendix
\section{Revenue Gap Examples}\label{appendix:example}

\begin{example}[Revenue Gap for Myerson Reserve]
\label{ex:rev.gap.myserson}
Let the underlying network be the clique $G = K_n$, and
the valuations drawn i.i.d.~from the exponential distribution 
$F(v) = 1-e^{-v}$.
First, we notice that for any price $p$, the only equilibrium is the
symmetric one:
Given any price $p$, each equilibrium $\Ts \in \nash_p$ satisfies
$T_i \prod_{j \neq i} F(T_j) = p$ for all $i$, 
so $\frac{F(T_i)}{T_i} = \frac{F(T_j)}{T_j}$ for any $i \neq j$.
The fact that $F$ is strictly concave implies that $\frac{F(T)}{T}$ is
strictly monotone decreasing. 
Therefore, $\frac{F(T_i)}{T_i} = \frac{F(T_j)}{T_j}$ 
implies that $T_i = T_j$. 
So for each price, the only equilibria are symmetric ones,
characterized by the equation $T \cdot F(T)^{n-1} = p$.

The exponential distribution has virtual value function 
$\phi(v) = v-1$; therefore, the Myerson Reserve Price is $r=1$.
For this price, the threshold $T$ applied by
the agents satisfies $1 = T(1 - e^{-T})^{n-1} =: g(T)$.  Note
that $g(T)$ is strictly increasing in $T$.  
We write $m = n-1$ for convenience.
Let $T_1 = \log m - \log \log \log m$.  Then,
\begin{align*}
g(T_1) & = (\log m - \log\log\log m)\left(1-\frac{\log\log m}{m}\right)^{m-1} \\
& < (\log m) \left(1-\frac{\log\log m}{m}\right)^{\frac{m}{\log\log m}\log\log
m}\left(1 - \frac{\log \log m}{m}\right)^{-1}  \\ 
& < (\log m) \cdot e^{-\log\log m} 
\; = \; 1.
\end{align*}
Thus, $T = g^{-1}(1)$ is greater than $T_1$, and hence, 
$F(T) > F(T_1) = 1 - \frac{\log \log (n-1)}{(n-1)}$.  We therefore
have $\rev(1) = \rev(1, T) = n \cdot (1-F(T)) < 2\log\log n$.

On the other hand, if we set $p = \log n \cdot (1 - 1/n)^{n-1}$, 
the corresponding threshold is $T' = g^{-1}(p) = \log n$.
We then have 
$\rev(p) = \rev(p, T') = p n \cdot (1 - F(T')) = p > \quarter \log n$.
The gap between the optimal revenue and the revenue of the Myerson reserve price
can therefore be as large as $\Omega(\frac{\log n}{\log \log n})$.
\end{example}

\begin{example}[Linear gap in revenue among equilibria]
\label{ex:rev.gap}
Consider a graph $G$ composed of a cycle of size $5$;
in addition, for every triple of $3$ out of the $5$ nodes in the
cycle, there is a set of $N$ nodes, each connected to each of
the nodes in the triple.  
The total number of nodes is $n = 10 N + 5$.

We choose $F$ to be the uniform distribution on $[0,1]$.  Pick an arbitrary
price $p \in (0,1)$.
We now consider two equilibria from $\nash_p$.  
In the first, each of the 5 buyers on the
cycle chooses a threshold of $p^{1/3}$, and each of the remaining $10
N$ buyers chooses a threshold of $1$.
Note that $T_i \cdot \prod_{j \in N(i)} F(T_j) = (p^{1/3})^3 = p$ 
for each buyer $i$, so this is indeed an equilibrium.  
The revenue at this equilibrium is $5 p (1 - p^{1/3} )$.

For the second equilibrium, two non-adjacent buyers on the cycle will
choose threshold $p$. 
The remaining three buyers on the cycle, say $S$, will choose threshold $1$. 
The $N$ buyers who are adjacent to the three nodes in $S$ will choose
threshold $p$; the remaining $9N$ buyers will choose threshold $1$.
Note that $T_i \prod_{j \in N(i)} F(T_j) \leq p$ for each buyer $i$. 
Moreover, for each buyer for whom $T_i \prod_{j \in N(i)} F(T_j) < p$
(i.e., each buyer with multiple neighbors who have threshold $p$),
we have $T_i = 1$. The proposed thresholds therefore form an
equilibrium.  The revenue at this equilibrium is
$(N+2) \cdot p \cdot ( 1 - p )$.

Taking $N$ to be arbitrarily large relative to $p$, we have that the
gap between the 
revenue of these two equilibria is $\Omega(N) = \Omega(n)$.
\end{example}

\section{Proof of Theorem \ref{thm:NP-hard-rev}}
\label{app:NP-hard}

In this section, we provide a proof of
Theorem~\ref{thm:NP-hard-rev}. 
We restate the theorem here for convenience.
\begin{rtheorem}{Theorem}{\ref{thm:NP-hard-rev}}
Assume that all agents' valuations are drawn i.i.d.~from the uniform
distribution on $[0,1]$.
Given a graph $G$ with $n$ nodes and a uniform price vector 
$\vc{p} = p\cdot \vc{1}$, it is NP-hard to approximate the worst-case
revenue to within a factor $n^{1-\epsilon}$.
\end{rtheorem}

\begin{proof}
We will use the characterization of the worst-case revenue with price
$p$ established in the proof of Theorem~\ref{thm:uniform_general_network}.
There, we showed that the revenue can be characterized by the
mathematical program~\eqref{eqn:LP}, repeated here for convenience:

\begin{LP}{Minimize}{%
p \cdot \sum_i (1-\exp( -x_i \cdot \log(1/p)))}
\sum_{j \in N(i) \cup \{i\}} x_j \geq 1 & \text{for all } i \\
\sum_{j \in N(i) \cup \{i\}} x_j > 1 \implies x_i = 0
& \text{for all } i \\
0 \leq x_i \leq 1 & \text{for all } i.
\end{LP}

To prove approximation hardness, we give a reduction from
$3$-SAT. More precisely, given a $3$-SAT instance, we construct an
instance of the pricing problem such that if a formula is satisfiable,
the revenue is $O(1)$, whereas if it is not satisfiable, the 
revenue is $\Omega(n^{1-\epsilon})$.

\begin{figure}[h]
\centering
\begin{tikzpicture}[scale=1.5]
\tikzstyle{blackdot}=[circle,draw=black,fill=black,thin,
inner sep=0pt,minimum size=1.5mm]
\tikzstyle{whitedot}=[circle,draw=black,fill=white,thin,
inner sep=0pt,minimum size=3.5mm]

\node (T) at (-0.5,0) [whitedot] {$\;\mathrm{T}\;$};
\node (F) at (0.5,0) [whitedot] {$\;\mathrm{F}\;$};
\node (T1) at (-1,0.3) [whitedot] {\phantom{$\;\mathrm{F}\;$}};
\node (T2) at (-1,-0.3) [whitedot] {\phantom{$\;\mathrm{F}\;$}};
\node (F1) at (1,0.3) [whitedot] {\phantom{$\;\mathrm{F}\;$}};
\node (F2) at (1,-0.3) [whitedot] {\phantom{$\;\mathrm{F}\;$}};
\draw (T) -- (F);
\draw (T) -- (T1);
\draw (T) -- (T2);
\draw (F) -- (F1);
\draw (F) -- (F2);
\end{tikzpicture}
\caption{A variable gadget.}
\label{fig:widget}
\end{figure}
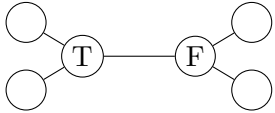

\paragraph{The reduction} Given a $3$-SAT formula with $k$ clauses and $m$
variables, construct an instance of the pricing problem as follows. 
For each variable in the formula, construct a ``gadget,'' as depicted
in Figure~\ref{fig:widget}; the nodes $\mathrm{T}$ and $\mathrm{F}$
will correspond to true and false assignments, respectively. For each clause,
we introduce $k^L$ nodes (for some fixed large $L$), each connected to three
other nodes, as follows. For every positive literal in the clause, each of
the $k^L$ clause nodes is connected to the $\mathrm{T}$-node of the
corresponding variable gadget, and for every negative literal, each of the 
$k^L$ clause nodes is connected to the $\mathrm{F}$-node of the
corresponding variable gadget

\paragraph{An observation on the resulting instance} 
We observe the following fact on the instance of the pricing problem
obtained from the reduction. 
Let $u$ and $v$ be the respective $\mathrm{T}$-node and
$\mathrm{F}$-node of a variable gadget. 
We claim that all feasible solutions to the program~\eqref{eqn:LP}
have $x_u, x_v \in \{0,1\}$. 
Indeed, let $w_1$ and $w_2$ be the leaf nodes of the variable gadget
attached to $u$. We know by the first constraint of the program that
$x_{w_1} \geq 1 - x_u$ and $x_{w_2} \geq 1 - x_u$. Therefore, 
$\sum_{i \in N(u) \cup \{u\}} x_i \geq x_u + 2\cdot (1-x_u) = 2 - x_u$.
If $x_u < 1$, then $\sum_{i \in N(u) \cup \{u\}} x_i > 1$, and the second
condition in \eqref{eqn:LP} implies that $x_u = 0$. 
The same analysis applies to $x_v$.

Also notice that $x_u = x_v = 1$ violates the second condition, since
$\sum_{i \in N(u) \cup \{u\}} x_i \geq 2 > 1$, but $x_u \neq 0$. Therefore, the
pair $(x_u, x_v)$ must be one of the set $\{(0,1), (1,0), (0,0)\}$.

We also observe that if $i$ is a clause node that is connected to a
variable node $u$ with $x_u = 1$, then $x_i = 0$ by the second condition.

\paragraph{Satisfiable formula implies low worst-case revenue} 
If the formula is satisfiable, consider a satisfying truth assignment. 
If a certain variable is assigned True, then set $x_u = 1$ for its
$\mathrm{T}$-node and $x_v = 0$ for its $\mathrm{F}$-node. 
Apply the opposite assignment for a variable that is assigned False.
For the leaf nodes of the variable gadget, set their values
to be one minus the value of their neighbor. 
Finally, set $x_v = 0$ for all all clause nodes $v$.
Since the assignment is satisfiable, each clause node is connected to
at least one node that is assigned $1$. 
Therefore, the worst-case revenue is at most 
$p \cdot (1-p) \cdot3 \cdot m$:
each node assigned $1$ produces revenue $p(1-p)$, and there
are exactly $3m$ such nodes.

\paragraph{Unsatisfiable formula implies high worst-case revenue} 
If the formula is not satisfiable,
then for every assignment of values $x_v$ to nodes that is
feasible for~\eqref{eqn:LP},
there exists at least one clause whose nodes are connected to three nodes
with value $0$. By the first condition of the program, all such clause
variables must be set to $1$. The revenue is therefore at least $3m + k^L$.

Since the graph has $n = 6m+ k \cdot k^L$ nodes, by setting $L$ to be
sufficiently large, one cannot distinguish between a solution of revenue
$n^{1-1/L}$ and a solution of revenue $O(1)$. 
This gives an $\Omega(n^{1-\epsilon})$ hardness
for approximating the worst-case revenue, for every $\epsilon >0$.
\end{proof}

\end{document}